\documentclass[journal,10pt]{IEEEtran}

\usepackage{amsmath,amssymb,amsthm}
\usepackage{booktabs}
\usepackage{cite}
\usepackage{microtype}
\usepackage[hidelinks]{hyperref}

\hypersetup{
  pdftitle={Compressed Single-Tone Frequency Estimation With Unknown Complex Gain: Singular Rates and Global Identifiability},
  pdfauthor={Armon Rasooli},
  pdfsubject={Signal processing and frequency estimation}
}

\newcommand{\C}{\mathbb C}
\newcommand{\T}{\mathbb T}
\newcommand{\CN}{\mathcal{CN}}
\newcommand{\Gr}{\operatorname{Gr}_{\C}}
\newcommand{\tr}{\operatorname{tr}}
\newcommand{\GI}{\mathrm{GI}}
\newcommand{\eps}{\varepsilon}

\newtheorem{theorem}{Theorem}
\newtheorem{lemma}[theorem]{Lemma}
\newtheorem{proposition}[theorem]{Proposition}
\newtheorem{corollary}[theorem]{Corollary}

\title{Compressed Single-Tone Frequency Estimation With Unknown Complex Gain:\\
Singular Rates and Global Identifiability}

\author{Armon Rasooli%
\thanks{This work received no external funding.}%
\thanks{Armon Rasooli is with Iran University of Science and Technology, Tehran, Iran
(e-mail: armonrasooli@gmail.com). ORCID:
\href{https://orcid.org/0009-0002-6583-7090}{0009-0002-6583-7090}.}}

\begin{document}
\maketitle
\thispagestyle{empty}
\pagestyle{empty}

\begin{abstract}
Fixed linear compression can preserve local Fisher information for a sinusoid yet destroy global
frequency identification, while a dark response can make local estimation nonregular. We study both
failures for a single complex tone observed through a fixed complex-linear sketch with unknown
nonzero complex gain and pre-sketch white Gaussian noise. Near an isolated analytic dark frequency,
we separate radial signal vanishing from optimized projective contact between the two signed
frequency branches. When the gain magnitude is constrained to a fixed nondegenerate interval,
finite contact is equivalent to local quotient identifiability and minimax consistency. The sharp
mean-square-error rate is determined by the sum of the radial and contact orders; infinite contact
produces exact local aliases. Globally, we formulate a worst-frequency efficient-information
objective on whitened row spaces. We solve it exactly for every even output rank, prove uniqueness
and quantitative rigidity of the symmetric-edge projector, and solve the even-aperture co-rank-one
case. The rank-two optimum is aliased, and a winding obstruction gives a positive lower bound on
the information price of global identification. From three outputs onward, sketches that are
globally identifying with an immersive projective response are open and dense and have zero price
at the level of suprema; for every even rank of at least four, the exact optimizer has this
property. Thus local information preservation, singular recoverability, and global identifiability
obey distinct compression laws within one estimation model.
\end{abstract}

\begin{IEEEkeywords}
Frequency estimation, data compression, statistical signal processing, global identifiability,
Fisher information, nonregular estimation, experiment design.
\end{IEEEkeywords}

\section{Introduction}
\IEEEPARstart{F}{requency} estimation is often designed locally: one selects measurements that
retain a large derivative of the mean, or equivalently a large Fisher information. Under an unknown
complex amplitude, however, the physically observable response is a complex line. Two failures are
then invisible to a purely regular, pointwise calculation. First, a compressed response can vanish
at one frequency; the first nonzero derivative need not determine the ensuing nonregular rate.
Second, two separated frequencies can produce proportional compressed responses and hence identical
noiseless data after changing the unknown gain. This paper gives quantitative laws for both
failures in one single-tone experiment.

Fisher-preserving compression and the distinction between Fisher equality and statistical
sufficiency are established topics. MOPED-type constructions preserve a local Fisher matrix at a
fiducial model, while their compressed likelihoods may retain degeneracies
\cite{Heavens2000,Graff2011}; Pollard gave a general score-based criterion showing that Fisher
information can be preserved by an insufficient statistic \cite{Pollard2013}. For nonlinear
Gaussian means, projector formulas for information after compression are also known
\cite{Pakrooh2015}. These antecedents already separate local information from global
identification. The problem here is to solve the resulting quantitative harmonic design.

Endpoint and clustered coordinate-sampling patterns are likewise antecedented
\cite{Oliphant2004,Trittler2009,VanderWerf2026}. Classical secant avoidance and
diagonal compactification explain generic injective projections
\cite{Laksov1976,Johnson1978,FultonMacPherson1994,Gorlach2019}. Those results do
not determine the optimum over all complex row spaces after gain profiling, its
projector-level equality cases, the information price of global identification,
or the singular rate at a dark response.

Our contributions are fourfold.
\begin{enumerate}
\item We prove a finite/infinite analytic-contact phase law at an isolated dark response. Under a
fixed strict gain annulus, finite optimized opposite-branch projective contact is equivalent to
local quotient identifiability and local signed minimax consistency; the full-pair exponent is the
sum of the radial and contact orders. Infinite contact yields admissible exact aliases.
\item We solve the arbitrary-complex-row-space information problem for every even output rank. The
symmetric-edge row-space projector is uniquely optimal and quantitatively rigid; for every even
rank at least four it is regularly globally identifying. We also solve the even-aperture
co-rank-one family exactly.
\item We quantify global-identification cost. At rank two, a winding-protected neighborhood of the
unique information optimum is aliased, yielding an explicit positive price floor. From rank three
onward, regularly globally identifying designs are open and dense and the constrained and
unconstrained values agree as suprema.
\item We identify the exact zero set of the uniform information functional and classify lossless
ranks. This produces a concrete lossless-but-aliased example and a compact local/global design
taxonomy.
\end{enumerate}

The local and global results share the same signal-processing cause: quotienting the unknown gain
converts both near-diagonal dark collisions and off-diagonal aliases into collisions of complex
lines. Analytic
contact controls the severity of the former; spectral rigidity, winding, and secant incidence
control the avoidability of the latter. The compactified incidence used below is classical
infrastructure rather than a claimed new geometric construction.

\begin{table*}[t]
\caption{Relation to the closest established lines of work. Each row separates an antecedented
mechanism from the theorem-level question addressed here.}
\label{tab:priorart}
\centering
\footnotesize
\setlength{\tabcolsep}{4.0pt}
\begin{tabular}{@{}p{0.18\textwidth}p{0.34\textwidth}p{0.41\textwidth}@{}}
\toprule
Line of work & Established result & Question not supplied by that result alone\\
\midrule
Fisher-preserving compression \cite{Heavens2000,Graff2011,Pollard2013,Pakrooh2015}
& Local Fisher matrices can be preserved under compression, even when the statistic is not
sufficient or the compressed likelihood has additional modes.
& These results do not determine the worst-frequency optimum over all complex harmonic row spaces,
its equality cases, or the information cost of requiring global frequency identification.\\[2pt]
Single-tone sampling design \cite{Oliphant2004,Trittler2009,VanderWerf2026}
& Endpoint concentration and clustered coordinate sampling are established design motifs.
& Coordinate-design results do not by themselves prove optimality over arbitrary complex row
spaces, projector rigidity, or the exact even-rank and co-rank-one classifications below.\\[2pt]
Nonregular estimation and optimal recovery
\cite{Donoho1994,Chernoyarov2018,HoNguyen2019,Plummer2026}
& Singular moduli, cusp orders, and observable orders can control nonstandard rates in their
respective models.
& They do not supply the present full-pair modulus under an annular complex nuisance or the
finite/infinite projective-contact identifiability equivalence.\\[2pt]
Analytic contact and preparation \cite{Tamm1981,Parusinski1994,Birbrair2017}
& Subanalytic preparation and arc-contact theory provide existence and arithmetic of leading
exponents.
& These tools do not establish that every nuisance-profiled pair stratum is no worse than \(r+c\),
or convert that exponent into the all-large-\(R\) minimax law below.\\[2pt]
Secant and projection geometry \cite{Laksov1976,Johnson1978,FultonMacPherson1994,Gorlach2019}
& Secant avoidance, tangent limits, and configuration compactification support generic embedding
arguments.
& Generic projection theory does not couple injectivity to the worst-frequency information
objective, prove the rank-two winding exclusion, or calculate the resulting price floor.\\
\bottomrule
\end{tabular}
\end{table*}

Table~\ref{tab:priorart} is deliberately mechanism-specific. Classical tools are essential to the
proofs, but the claims below concern their quantitative interaction in the fixed harmonic
experiment rather than the novelty of the tools themselves.

\section{Experiment and Quotient Geometry}
The underlying complex-tone likelihood and its classical regular information analysis date at least
to \cite{RifeBoorstyn1974}. Let
\begin{equation}
a_L(\omega)=(1,e^{i\omega},\ldots,e^{i(L-1)\omega})^{\mathsf T},
\qquad \omega\in\T,
\end{equation}
where \(L\ge2\). For independent repetitions,
\begin{equation}
X_j=\alpha a_L(\omega)+W_j,\qquad
W_j\sim\CN(0,\sigma^2I_L),
\label{eq:rawmodel}
\end{equation}
with unknown \(\alpha\in\C^\times\). A fixed full-row-rank \(B\in\C^{m\times L}\) gives \(BX_j\).
Exact output whitening replaces \(B\) by
\begin{equation}
C=(BB^*)^{-1/2}B,\qquad CC^*=I_m,\qquad P=C^*C.
\end{equation}
Thus \(Y_j=CX_j\sim\CN(\alpha Ca_L(\omega),\sigma^2I_m)\), and every row-space result is a
statement about the rank-\(m\) orthoprojector \(P\). This reduction assumes pre-sketch white noise
and arbitrary invertible output calibration. It is not a passive-hardware or post-combiner-noise
theorem.

At a nonzero response \(h(\omega)=Ca_L(\omega)\), the nuisance gain leaves only
\([h(\omega)]\in\mathbb{CP}^{m-1}\). For nonzero \(x,y\), direct least squares gives
\begin{equation}
\inf_{\beta\in\C}\|\alpha x-\beta y\|
=|\alpha|\,\|x\|\,d_{\rm ch}([x],[y]),
\label{eq:profileidentity}
\end{equation}
where \(d_{\rm ch}([x],[y])^2=1-|x^*y|^2/(\|x\|^2\|y\|^2)\). Projectivization fails exactly at a
dark response \(h(\omega_0)=0\); Section~III supplies the replacement geometry.

\subsection{Efficient Score After Gain Profiling}
The normalization used below follows directly from the declared Gaussian experiment. Write
\(\alpha=\alpha_{\rm R}+i\alpha_{\rm I}\), \(h(\omega)=Ca_L(\omega)\), and
\(\nu(\omega,\alpha)=\alpha h(\omega)\). For one observation,
\begin{equation}
\ell(\omega,\alpha;y)=\mathrm{const}-\frac{\|y-\nu(\omega,\alpha)\|^2}{\sigma^2}.
\end{equation}
For real coordinates \(\vartheta=(\omega,\alpha_{\rm R},\alpha_{\rm I})\), the \(R\)-sample Fisher
matrix satisfies
\begin{equation}
[\mathcal I_R(\vartheta)]_{uv}=\frac{2R}{\sigma^2}
\operatorname{Re}\!\left\{(\partial_u\nu)^*(\partial_v\nu)\right\}.
\label{eq:realfisher}
\end{equation}
The two gain derivatives are \(h\) and \(ih\); their real span is the complex line
\(\operatorname{span}_{\C}\{h\}\). At \(h\ne0\), the Schur complement of this nuisance block is
\begin{equation}
\mathcal I_{\omega\omega}^{\rm eff}
=\frac{2R|\alpha|^2}{\sigma^2}
\inf_{\lambda\in\C}\|h'-\lambda h\|^2
=\frac{2R|\alpha|^2}{\sigma^2}\|P_h^\perp h'\|^2.
\end{equation}
Thus the deterministic projective conorm and the Gaussian efficient score coincide only after the
probability model and covariance have been fixed.

For the regular information problem, put
\begin{equation}
\mu=\frac{L-1}{2},\quad
D=\operatorname{diag}(n-\mu)_{n=0}^{L-1},\quad
g(\omega)=iDa_L(\omega),
\end{equation}
and
\begin{equation}
S_L=\|g(\omega)\|^2=\frac{L(L^2-1)}{12}.
\end{equation}
Indeed, for \(N=\operatorname{diag}(0,\ldots,L-1)=D+\mu I_L\),
\begin{equation}
a_L'=iNa_L=g+i\mu a_L,
\qquad h'=Cg+i\mu h.
\end{equation}
The centered term lies in the gain-nuisance line. Since \(\|Cv\|^2=v^*Pv=\|Pv\|^2\), define
\begin{align}
\rho_P(\omega)&=\frac1{S_L}\inf_{\lambda\in\C}
 \|P(g(\omega)-\lambda a_L(\omega))\|^2,\label{eq:rho}\\
F(P)&=\inf_{\omega\in\T}\rho_P(\omega),\qquad
\Gamma_{L,m}=\sup_{\operatorname{rank} P=m}F(P).\label{eq:Fgamma}
\end{align}
The supremum in \eqref{eq:Fgamma} is attained by the continuity proved with
Theorem~\ref{thm:global}; the frequency operation remains an infimum. Up to
\(2R|\alpha|^2S_L/\sigma^2\), \(\rho_P\) is the efficient Fisher information for frequency after
profiling the complex gain.
For full observation, \(a_L^*g=i\sum_n(n-\mu)=0\), so \(S_L\) is exactly the full-aperture
efficient benchmark and \(F(P)\) is its worst-frequency retained fraction. At darkness the nuisance
block in \eqref{eq:realfisher} is singular; the Schur-complement interpretation is unavailable,
which is why pointwise and design-level continuity are separated below.

A projector is globally identifying (GI) if \(Pa_L(\omega)\ne0\) for every \(\omega\) and
\(\omega\mapsto[Pa_L(\omega)]\) is injective. It is regular-GI if this map is also immersive. Write
\begin{equation}
\Gamma^{\GI}_{L,m}=\sup_{P\in\mathfrak G^{\GI}_{L,m}}F(P),
\qquad \Pi_{L,m}=\Gamma_{L,m}-\Gamma^{\GI}_{L,m}.
\end{equation}

\section{Dark Responses: Contact and Minimax Rate}
Fix a dark frequency and use the signed physical chart \(t=\omega-\omega_0\). The whitened analytic
response factors uniquely as
\begin{equation}
h(t)=t^rq(t),\qquad r\in\mathbb N,\quad r\ge1,\quad q(0)\ne0,
\label{eq:germ}
\end{equation}
after shrinking the chart so \(q(t)\ne0\). The integer \(r\) is the radial order: it controls how
rapidly signal magnitude disappears. It is invariant under a local analytic source coordinate with
nonzero derivative and under a constant invertible output map.

To measure what remains after gain profiling, define for small \(\rho>0\)
\begin{equation}
\delta(\rho)=\min_{\rho\le t,s\le2\rho}
d_{\rm ch}([q(t)],[q(-s)]).
\end{equation}
Restricted analytic data are globally subanalytic on a sufficiently small compact chart. Proper
minimization and one-variable preparation \cite{Tamm1981,Parusinski1994} imply either \(\delta\) is
eventually zero, denoted \(c=\infty\), or
\begin{equation}
\delta(\rho)=C_\delta\rho^c+o(\rho^c),\qquad C_\delta>0,\quad c\in\mathbb Q_{>0}.
\end{equation}
Supplementary Lemma~8 proves that the finite exponent and the eventual-zero alternative are
independent of the fixed window factor \(2>1\). The contact order \(c\) quantifies how rapidly the
positive and negative frequency branches become indistinguishable as complex lines after the
unknown gain has been removed.

\subsection{Why the Two Orders Add}
The roles of \(r\) and \(c\) differ. On a balanced opposite-branch pair \(t,s\asymp\rho\),
\begin{equation}
\|h(t)\|\asymp\rho^r,\qquad
d_{\rm ch}([q(t)],[q(-s)])\asymp\rho^c,\qquad t+s\asymp\rho.
\end{equation}
Equation~\eqref{eq:profileidentity} therefore predicts a separation of order \(\rho^{r+c}\): radial
vanishing reduces signal magnitude, while projective contact reduces the directional distinction
surviving gain profiling. This is only a candidate worst branch. Theorem~\ref{thm:contact}
additionally proves that imbalanced, same-side, and gain-boundary sequences cannot vanish faster.

Choose an affine chart at \([q(0)]\) and write the first nonconstant projective jet as
\begin{equation}
z(t)=z(0)+v_{k_0}t^{k_0}+O(t^{k_0+1}),\qquad v_{k_0}\ne0.
\end{equation}
Equal radii show \(c\ge k_0\). If \(k_0\) is odd, \(z(t)-z(-s)\) has leading term
\(v_{k_0}(t^{k_0}+s^{k_0})\), so cancellation is impossible for positive \(t,s\) and \(c=k_0\).
Contact exceeding \(k_0\) forces even \(k_0\) and \(s/t\to1\) on every sharper arc. On a same-side
balanced cone, the same jet yields
\begin{equation}
\|z(t)-z(u)\|\ge c_0M^{k_0-1}|t-u|,\qquad M=\max\{|t|,|u|\}.
\end{equation}
Since \(k_0\le c\), that stratum is no worse than exponent \(r+c\). The family
\begin{equation}
q_j(t)=(1,t^2,t^{2j+1})^{\mathsf T},\qquad j\ge1,
\end{equation}
has fixed first projective order two but contact \(c=2j+1\). Thus neither radial order nor the
first visible projective derivative alone determines the singular rate.
Indeed, on \(t,s\in[\rho,2\rho]\) its affine opposite-branch difference contains
\(t^{2j+1}+s^{2j+1}\ge2\rho^{2j+1}\), while equal radii cancel the quadratic coordinate and attain
that order.

\begin{theorem}[Analytic contact phase law]\label{thm:contact}
Let \eqref{eq:germ} be a nonzero real-analytic dark germ and fix, before taking any local limit,
the strict annulus
\(
\mathcal A=\{\alpha\in\C:A_-\le|\alpha|\le A_+\}
\)
with \(0<A_-<A_+<\infty\). There is \(\eps_0>0\) such that the following are equivalent for every
fixed \(0<\eps\le\eps_0\):
(i) \(c<\infty\); and
(ii) \(\alpha h(t)=\beta h(u)\), with \(t,u\in[-\eps,\eps]\) and \(\alpha,\beta\in\mathcal A\),
implies \(t=u\).

If \(c<\infty\) and \(p=r+c\), then a constant \(C_\eps>0\) satisfies
\begin{equation}
\|\alpha h(t)-\beta h(u)\|\ge C_\eps|t-u|^p
\label{eq:fullpair}
\end{equation}
for all such \(t,u,\alpha,\beta\). The exponent is sharp, and the all-scale opposite-branch profile
obeys
\begin{equation}
\begin{aligned}
\Delta_{\mathcal A}(\rho)
&=\min_{\substack{\rho\le t,s\le2\rho\\\alpha,\beta\in\mathcal A}}
\|\alpha h(t)-\beta h(-s)\|\\
&=C_{\mathcal A}\rho^p+o(\rho^p),\qquad C_{\mathcal A}>0.
\end{aligned}
\end{equation}
If \(c=\infty\), every sufficiently small fixed interval contains positive \(t,s\) and admissible
gains satisfying \(\alpha h(t)=\beta h(-s)\).
\end{theorem}

\emph{Proof mechanism.}
Supplementary Lemmas~8--9 give the proof. A priority cover separates central, radially
imbalanced, balanced same-side, and balanced opposite-side pairs. Radial imbalance is controlled by
norm separation. Same-side pairs are controlled by the first nonconstant projective jet; its order
\(k_0\) satisfies \(k_0\le c\). Only balanced opposite-side pairs attain the contact exponent. A
definable minimizing arc supplies sharpness, while its least-squares gain satisfies
\(|\beta_*|/|\alpha|\to1\), making the sharp competitor admissible for an interior true gain.
Infinite contact gives an exact projective collision with \(s/t\to1\) and the same
gain-admissibility limit. This full-pair step, rather than preparation alone, produces \(r+c\).

\emph{Quantifiers and uniformity.}
First the analytic germ, \(A_-\), \(A_+\), and noise scale are fixed. Then
\begin{equation}
\exists\eps_0>0\ \ \forall\,0<\eps\le\eps_0\ \ \exists C_\eps>0\ \
\forall(t,u,\alpha,\beta)\in I_\eps^2\times\mathcal A^2,
\end{equation}
the bound \eqref{eq:fullpair} holds. Its constant is uniform over both gains and all pair strata,
but no uniformity is asserted as \(A_-\downarrow0\), \(A_+\downarrow A_-\), or the neighborhood
changes. Sharpness is existential: an interior true gain and annulus-admissible opposite-branch
competitors attain order \(r+c\). In the infinite branch, the positive risk floor is attached to
each fixed neighborhood and need not be uniform as that neighborhood shrinks.

\begin{corollary}[Local signed minimax law]
In the repeated whitened experiment
\(Y_j\sim\CN(\alpha h(t),\sigma^2I_m)\), \(j=1,\ldots,R\),
fix \(\sigma>0\), the strict annulus above, and a sufficiently small signed interval
\(I_\eps=[-\eps,\eps]\). For \(c<\infty\),
\begin{equation}
\inf_{\widehat t}\sup_{t\in I_\eps,\,\alpha\in\mathcal A}
\mathbb E_{t,\alpha}(\widehat t-t)^2
=\Theta\!\left(R^{-1/(r+c)}\right).
\label{eq:minimax}
\end{equation}
Thus localization and root-mean-square error scale as \(R^{-1/[2(r+c)]}\). For \(c=\infty\), the
minimax risk in every fixed sufficiently small interval is bounded away from zero for all \(R\).
\end{corollary}

Let \(p=r+c\). The all-scale profile, rather than one sharp sequence, makes the lower bound valid
for every sufficiently large integer \(R\). Choose \(\rho_R=aR^{-1/(2p)}\) and an attained
profile-minimizing pair. For a constant \(M_0\) independent of \(R\),
\begin{equation}
2\rho_R\le d_R=t_R+s_R\le4\rho_R,\qquad
\|\nu_{0,R}-\nu_{1,R}\|\le M_0\rho_R^p.
\end{equation}
Under the convention in \eqref{eq:rawmodel},
\begin{equation}
\operatorname{KL}\!\left(P_{\nu_{0,R}}^{\otimes R}\middle\|P_{\nu_{1,R}}^{\otimes R}\right)
=\frac{R\|\nu_{0,R}-\nu_{1,R}\|^2}{\sigma^2}
\le\frac{M_0^2a^{2p}}{\sigma^2}.
\end{equation}
Taking \(a\) small makes total variation at most a fixed \(\eta<1\), and the two-point inequality
gives
\begin{equation}
\mathcal R_R(\eps)\ge\frac{d_R^2}{8}(1-\eta)\ge c_{\eps,-}R^{-1/p}.
\label{eq:lecamrate}
\end{equation}
For the upper bound, \(\bar Y=\alpha h(t)+\bar W\), where \(\bar W\sim\CN(0,\sigma^2I_m/R)\). A
measurable minimum-distance selector on \(I_\eps\times\mathcal A\) satisfies
\begin{equation}
\begin{aligned}
\|\widehat\alpha h(\widehat t)-\alpha h(t)\|&\le2\|\bar W\|,\\
|\widehat t-t|^2&\le(2/C_\eps)^{2/p}\|\bar W\|^{2/p}.
\end{aligned}
\end{equation}
For \(Z\sim\CN(0,I_m)\), \(\|Z\|^2\sim\operatorname{Gamma}(m,1)\), hence
\begin{equation}
\mathbb E\|\bar W\|^{2/p}
=\left(\frac{\sigma^2}{R}\right)^{1/p}\frac{\Gamma(m+1/p)}{\Gamma(m)}.
\label{eq:complexmoment}
\end{equation}
Equations~\eqref{eq:lecamrate}--\eqref{eq:complexmoment} prove \eqref{eq:minimax}; exact aliases
instead give \(\mathcal R_R(\eps)\ge d_\eps^2/4\). Supplementary Lemma~10 supplies the selection
and constant details. This is a local signed result, and only the chosen binary subexperiment is
assigned the corresponding Gaussian limit. Nonregular rates and modulus methods have broad
antecedents \cite{Chernoyarov2018,HoNguyen2019,Donoho1994}; none supplies the annular full-pair
quotient law above.

The gain annulus is load-bearing. For
\begin{equation}
h(t)=t^r(1,t^2+t^3,t^5)^{\mathsf T},
\end{equation}
the strict-annulus exponent is \(r+5\), whereas fixed gain magnitude gives \(r+3\). At equal radius
the correct gain relation is \(\beta=(-1)^r\alpha\). Hence the annular theorem is not uniform as
\(A_+\downarrow A_-\).

\section{Uniform Information and the Local Discriminant}
At a non-dark response, evaluating the scalar least-squares problem in \eqref{eq:rho} gives
\begin{equation}
\rho_P(\omega)=\frac1{S_L}\left(\|Pg\|^2-
\frac{|(Pa)^*Pg|^2}{\|Pa\|^2}\right).
\end{equation}
This is positive precisely when the projected curve is projectively regular at \(\omega\). Darkness
requires a different limit and can make the frequency infimum unattained.

\begin{proposition}[Exact local discriminant]\label{prop:zeroset}
For every rank-\(m\) projector,
\begin{equation}
\begin{aligned}
F(P)>0\quad\Longleftrightarrow\quad
&Pa_L(\omega)\ne0\quad\forall\omega,\\
&\omega\mapsto[Pa_L(\omega)]\ \text{is immersive}.
\end{aligned}
\end{equation}
Consequently every design lies in exactly one class: \(F=0\) (local degeneracy), \(F>0\) with a
global alias, or regular-GI.
\end{proposition}

For a dark germ \(Pa_L(\omega_0+t)=t^rq(t)\), the identity \(Pg=h'-i\mu h\) and profiler
\(\lambda_t=r/t\) make the residual tend to zero; hence \(F=0\), even when the algebraic value of
\(\rho_P(\omega_0)\) is positive. Conversely, a non-dark immersion has a continuous positive
profile on the compact circle. Supplementary Lemma~15 supplies the full proof and also proves
continuity of \(F\) on the Grassmannian. Pointwise \(\rho_P(\omega)\) need not be continuous
through darkness; only the design functional \(F\) is asserted to be globally continuous.

\section{Exact Spectral Design}
The trace ceiling follows by taking \(\lambda=0\) in \eqref{eq:rho} and averaging over frequency:
\begin{equation}
F(P)\le\frac{\tr(PD^2)}{S_L}.
\end{equation}
For even \(m=2k\le L\), define the symmetric-edge set and its projector
\begin{equation}
\begin{aligned}
E_{L,m}&=\{0,\ldots,k-1\}\cup\{L-k,\ldots,L-1\},\\
Q_{L,m}&=P_{E_{L,m}}.
\end{aligned}
\end{equation}

\begin{theorem}[All-even spectral extremality]\label{thm:even}
For \(L\ge2\) and every even \(m\) with \(2\le m\le L\),
\begin{equation}
\Gamma_{L,m}=
\frac{m(m^2-3Lm+3L^2-1)}{L(L^2-1)}.
\label{eq:evenformula}
\end{equation}
The unique maximizing row-space projector is \(Q_{L,m}\). If \(m<L\), every rank-\(m\) projector
satisfies the rigidity bound
\begin{equation}
F(P)\le\Gamma_{L,m}-\frac{L-m}{S_L}\tr\!\big(P(I-Q_{L,m})\big),
\label{eq:rigidity}
\end{equation}
where \(\tr(P(I-Q))=\|P-Q\|_F^2/2\). For every even \(m\ge4\), \(Q_{L,m}\) is regular-GI and
uniquely attains \(\Gamma^{\GI}_{L,m}=\Gamma_{L,m}\).
\end{theorem}

The trace ceiling contains two inequalities:
\begin{equation}
\begin{aligned}
F(P)&\le\frac1{2\pi}\int_0^{2\pi}\rho_P(\omega)d\omega\\
&\le\frac1{2\pi S_L}\int_0^{2\pi}\|Pg(\omega)\|^2d\omega
=\frac{\tr(PD^2)}{S_L}.
\end{aligned}
\label{eq:averagetracechain}
\end{equation}
The first compares worst frequency with average; the second tests \(\lambda=0\). For \(m=2k\), the
top entries of \(D^2\) occur in the complete symmetric pairs indexed by \(E_{L,m}\). Ky Fan's
principle \cite{KyFan1949} selects their span. At \(Q=Q_{L,m}\),
\begin{equation}
(Qa_L)^*Qg=i\sum_{n\in E_{L,m}}(n-\mu)=0,
\end{equation}
so the profiler is already orthogonal and both inequalities in \eqref{eq:averagetracechain} are
equalities at every frequency. Exact summation gives
\begin{equation}
F(Q)=\frac{2}{S_L}\sum_{j=0}^{k-1}(\mu-j)^2
=\frac{m(m^2-3Lm+3L^2-1)}{L(L^2-1)}.
\end{equation}
When \(m<L\), the selected/unselected spectral gap is exactly \(L-m\), and every competing
rank-\(m\) projector obeys
\begin{equation}
\begin{aligned}
\tr(QD^2)-\tr(PD^2)&\ge(L-m)\tr(P(I-Q))\\
&=\frac{L-m}{2}\|P-Q\|_F^2.
\end{aligned}
\end{equation}
This proves rigidity and shows that uniqueness is at the row-space projector level: invertible
output-coordinate changes of \(B\) represent the same whitened design. For \(m\ge4\), \(Q\) retains
coordinates \(0,1\); their ratio \(e^{i\omega}\) certifies non-darkness, injectivity, and
immersion. Supplementary Lemma~11 contains the equality details.

The even theorem leaves an odd-rank boundary that is nevertheless exactly solvable.
\begin{proposition}[Even-aperture co-rank one]\label{prop:corank}
If \(L\ge4\) is even and \(m=L-1\), then
\begin{equation}
\Gamma_{L,L-1}=1-\frac{3}{(L-1)^2(L+1)}.
\label{eq:corank}
\end{equation}
Exactly two row-space projectors maximize: those deleting \(e_{L/2-1}\) or \(e_{L/2}\). Both are
regular-GI; every nontrivial mixture of the two central kernel vectors is strictly worse.
\end{proposition}
Write \(P=I-vv^*\), \(\|v\|=1\), and set \(A(\omega)=v^*a_L(\omega)\), \(G(\omega)=v^*g(\omega)\).
Dark designs have \(F=0\) and cannot maximize. Exact scalar profiling gives
\begin{equation}
F(I-vv^*)=1-\frac{L}{S_L}\sup_{\omega\in\T}
\frac{|G(\omega)|^2}{L-|A(\omega)|^2}.
\end{equation}
The supremum is at least the ratio of Haar integrals, so
\begin{equation}
\sup_\omega\frac{|G(\omega)|^2}{L-|A(\omega)|^2}
\ge\frac{v^*D^2v}{L-1}\ge\frac1{4(L-1)}.
\end{equation}
For even \(L\), equality forces \(v\) into the two central coordinates. Write
\(v=c_0e_{L/2-1}+c_1e_{L/2}\) and \(a_0=2|c_0c_1|\). Direct maximization gives
\begin{equation}
\sup_\omega\frac{|G(\omega)|^2}{L-|A(\omega)|^2}
=\frac{1+a_0}{4(L-1+a_0)},
\end{equation}
whose excess over \(1/[4(L-1)]\) is
\begin{equation}
\frac{a_0(L-2)}{4(L-1)(L-1+a_0)}.
\end{equation}
It vanishes exactly when \(c_0c_1=0\), proving both the value and the two-projector equality
classification. See Supplementary Proposition~12.

\begin{corollary}[Lossless ranks]
For \(L\ge2\),
\begin{equation}
\Gamma_{L,m}=1
\quad\Longleftrightarrow\quad
m=L\ \text{or}\ (L\ \text{is odd and }m=L-1).
\end{equation}
The proper lossless row space is the unique central deletion. It is regular-GI for odd \(L\ge5\).
At \(L=3,m=2\), it has response \([1:e^{2i\omega}]\): non-dark and immersive, but two-to-one.
\end{corollary}

Indeed, equality forces \(Pg(\omega)=g(\omega)\) for all \(\omega\); Fourier uniqueness shows that
the tangent span has dimension \(L\) for even \(L\) and \(L-1\) for odd \(L\). The \(L=3\) case
makes the local/global distinction exact: zero local information loss does not guarantee global
identification.

\emph{Boundary checks.}
With one output, every non-dark projective response lies in \(\mathbb{CP}^0\), so profiling removes
the scalar tangent; darkness also gives \(F=0\). Thus rank one supplies neither immersion nor GI.
For \(L=2,m=2\), full observation has response \([1:e^{i\omega}]\), \(F=1\), and regular-GI; the
rank-two alias and winding statements therefore begin at \(L\ge3\). Full rank always gives the
unique design \(P=I\) and \(F=1\). For odd \(L\), the proper lossless rank is \(m=L-1\), regular-GI
from \(L=5\), while \(L=3\) is the two-to-one exception. For even \(L\), no proper lossless rank
exists, but Proposition~\ref{prop:corank} solves the nearest rank \(L-1\).

\section{Global-Identification Price and Output Rank}
At rank two, Theorem~\ref{thm:even} gives
\begin{equation}
\Gamma_{L,2}=\frac{6(L-1)}{L(L+1)},
\end{equation}
uniquely at the endpoint plane. Its projective response is \([1:e^{i(L-1)\omega}]\), so it is
aliased for every \(L\ge3\). This does not mean two outputs cannot identify frequency: any
adjacent-coordinate plane has response \([1:e^{i\omega}]\) and \(F=1/(2S_L)>0\).

\begin{theorem}[Global-identification price and density]\label{thm:global}
For every \(L\ge3\),
\begin{equation}
\Pi_{L,2}\ge
\frac{3(L-2)(4L-1)}{L^3(L^2-1)}>0.
\label{eq:pricefloor}
\end{equation}
For every \(3\le m\le L\), regular-GI row spaces are open and dense and
\begin{equation}
\Gamma^{\GI}_{L,m}=\Gamma_{L,m}
\label{eq:suprema}
\end{equation}
as an equality of suprema. Exact attainment follows in the families proved above---full rank, every
even \(m\ge4\), odd-\(L\) central deletion, and the even-\(L\) co-rank-one case---but is not
asserted for unresolved proper odd ranks.
\end{theorem}

The two mechanisms are different. For \(m=2\), align a row-space frame with the endpoint frame.
Inside \(\sin^2\theta_{\max}<(4L-1)/(4L^2)\), its two response polynomials are uniformly close to
\(1\) and \(z^{L-1}\). Rouch\'e's theorem preserves their disk-zero counts; their ratio therefore
has winding \(L-1\). An injective loop in \(\C^\times\) has index only \(0,\pm1\): a simple closed
curve bounds a Jordan domain, and the origin has index zero outside and \(\pm1\) inside. Thus this
chamber is non-GI. Combining its radius with \eqref{eq:rigidity} gives \eqref{eq:pricefloor};
Supplementary Lemma~13 gives the constants. The inequality is a sufficient strict chamber and a
price floor, not the exact constrained optimum.

For \(m\ge3\), compactify the physical two-frequency pair so secants extend to tangents. The
polynomial divided difference
\begin{equation}
b_n(z,w)=\frac{w^n-z^n}{w-z}=\sum_{j=0}^{n-1}w^{n-1-j}z^j,\qquad b_0=0,
\end{equation}
has \(b(z,z)=\partial_za(z)\). Let \(K=\ker P\), \(\dim_\C K=k_{\rm ker}=L-m\). For fixed
\(S=\operatorname{span}\{a(z),b(z,w)\}\), the Schubert condition \(K\cap S\ne0\) has complex
codimension \(m-1\); the deeper stratum \(S\subseteq K\) is smaller. Since
\begin{equation}
\dim_{\mathbb R}\Gr(k_{\rm ker},L)=2k_{\rm ker}m,
\end{equation}
adding the two real physical frequency variables yields
\begin{equation}
\dim_{\mathbb R}\mathcal Z^\circ\le2k_{\rm ker}m-2(m-1)+2
=2k_{\rm ker}m-2m+4.
\end{equation}
The design projection therefore has real codimension at least \(2m-4\); the one-real-dimensional
tangent boundary gives at least \(2m-3\). Using the physical two-real-dimensional pair is
essential: all complex secants answer a different design question. At \(m=3\), the bounds are two
and three, so the compact projected bad set has empty interior. At \(m=2\), the interior count
gives no exclusion, consistently with the open winding-protected bad chamber. Supplementary
Lemma~14 treats every Schubert stratum and proves that the complement is precisely regular-GI. The
closure is joint in design and pair; it does not assert that every critical point of one fixed
design is approached by aliases of that same design.

Finally, continuity of \(F\) requires a bounded/divergent-profiler split. If near-minimizing
\(\lambda_j\) stays bounded, ordinary continuity applies. If \( |\lambda_j|\to\infty\) while
residuals stay bounded, then \(P_ja(\omega_j)\to0\), the limit design is dark, and
Proposition~\ref{prop:zeroset} gives \(F=0\). Approximating a Grassmannian maximizer by dense
regular-GI designs then yields \eqref{eq:suprema}. Bare GI is not claimed open.

\begin{table}[t]
\caption{Output-rank landmarks \((L\ge3)\)}
\label{tab:phase}
\centering
\footnotesize
\setlength{\tabcolsep}{3.2pt}
\begin{tabular}{@{}p{0.15\columnwidth}p{0.26\columnwidth}p{0.47\columnwidth}@{}}
\toprule
Rank & Information law & Global-identification conclusion\\
\midrule
\(m=2\) & Exact \(\Gamma_{L,2}\) & GI exists, but price has the positive floor
\eqref{eq:pricefloor}; exact constrained optimum open.\\
\(m=3\) & General \(\Gamma\) not closed-form & Regular-GI dense; zero price as suprema. General
attainment open.\\
Even \(m\ge4\) & Exact \eqref{eq:evenformula} & Unique optimizer is regular-GI; exact attainment.\\
Even \(L\), \(m=L-1\) & Exact \eqref{eq:corank} & Two regular-GI central-deletion optimizers.\\
\bottomrule
\end{tabular}
\end{table}

\section{Discussion and Limits}
The local result is a phase law: finite opposite-projective contact is equivalent to local quotient
identifiability and consistency, whereas infinite contact
creates exact aliases. Radial order controls signal magnitude; optimized projective contact
controls how rapidly the signed branches separate after gain profiling. Neither order alone
suffices.

The global result separates three resource landmarks (Table~\ref{tab:phase}). Two outputs can
identify frequency, yet a neighborhood of the information optimum cannot. Three outputs make the
regular-GI constraint free only at the level of suprema. Four outputs give a universal exact
attaining design because the even spectral optimizer retains an adjacent pair; this is not a
theorem that four is the first attaining rank for every fixed \(L\).

Within the even-rank family, the exact marginal value of adding the next symmetric output pair is
\begin{equation}
\Gamma_{L,m+2}-\Gamma_{L,m}
=\frac{6(L-m-1)^2}{L(L^2-1)},\qquad m+2\le L.
\end{equation}
For \(m\ge3\), continuity, density, and openness also imply an operational strengthening of the
supremum statement: for every \(\eta>0\), some nonempty open set of regular-GI designs satisfies
\(F>\Gamma_{L,m}-\eta\). This does not imply that a maximizer belongs to that set.

Open problems include the exact rank-two GI-constrained optimum, arbitrary-row-space values and
attainment for proper residual odd ranks, the complete singular local limit experiment, exact
minimax constants after fixing a leading normal form, and the transition as the gain annulus
collapses. Passive hardware constraints and alternative noise placements require a different
resource model.

\section{Conclusion}
Under unknown complex gain, fixed compression has two distinct failure modes. Local singular
collision is governed by radial vanishing and optimized opposite-branch projective contact; global
aliasing is governed by the interaction of spectral optimizer geometry and projective incidence.
This yields a sharp local minimax phase law, exact even-rank and co-rank-one information designs, a
positive two-output global-identification price floor, equality of constrained and unconstrained
suprema from three outputs, and exact regular-GI attainment for the proved even-rank families.

\clearpage
\section*{Supplementary Material}
\setcounter{theorem}{7}
\setcounter{equation}{53}
\setcounter{section}{0}
\renewcommand{\thesection}{S\Roman{section}}
\renewcommand{\thesubsection}{\thesection.\Alph{subsection}}
\renewcommand{\theHsection}{supp.\arabic{section}}
This supplement proves every model-specific step referenced by the main manuscript. It uses the
main paper's notation and continues its theorem and equation numbering. The only imported facts are
stated where used: proper subanalytic minimization and one-variable preparation, measurable
selection of extrema, Ky Fan's trace principle, Rouch\'e's theorem and the Jordan index fact, and
semialgebraic projection dimension. No numerical check is used as proof.

\emph{Imported results in the precise form used.}
We use the following statements only. A proper fibrewise minimum of a continuous globally
subanalytic function over nonempty compact fibres is attained and globally subanalytic; a positive
one-variable globally subanalytic germ has a rational Puiseux leading power; and nonempty definable
fibres admit definable selections \cite{Tamm1981,Parusinski1994,Valette2025}. If \(X\) is standard
Borel, \(\Theta\) is compact metric, and \(f:X\times\Theta\to\mathbb R\) is Borel in \(x\) and
continuous in \(\theta\), its nonempty compact argmin correspondence has a Borel selector
\cite{BrownPurves1973}. If a Hermitian matrix has ordered eigenvalues
\(\lambda_1\ge\cdots\ge\lambda_L\), then
\(\max_{\operatorname{rank} P=m}\tr(PH)=\sum_{j=1}^m\lambda_j\); a strict cutoff
\(\lambda_m>\lambda_{m+1}\) makes the top spectral projector the unique maximizer \cite{KyFan1949}.
Rouch\'e is used only for polynomials holomorphic near the closed disk: \(|g|<|f|\) on \(S^1\)
preserves the number of disk zeros, counted with multiplicity, and excludes boundary zeros
\cite{Conway1978}. Finally, semialgebraic images are semialgebraic, semialgebraic maps do not
increase dimension, and projections of compact sets are compact \cite{Bochnak1998}. Every local
analytic function below is first restricted to a compact chart, as required by the subanalytic
statements.

\section{Analytic Contact and the Full-Pair Law}

\begin{lemma}[Contact preparation and annular alias dichotomy]\label{lem:supp-contact}
Let \(h(t)=t^rq(t)\), \(r\ge1\), where \(q\) is real analytic and \(q(0)\ne0\). For any fixed
\(\chi>1\), set
\begin{equation}
\delta_\chi(\rho)=\min_{\rho\le t,s\le \chi\rho}
d_{\rm ch}([q(t)],[q(-s)]).
\label{eq:supp-deltachi}
\end{equation}
After restriction to a compact analytic chart, either every \(\delta_\chi\) is eventually zero,
denoted \(c=\infty\), or all have the same finite rational leading exponent \(c\). In the zero
case, a fixed strict annulus contains gains producing exact opposite-branch mean collisions
arbitrarily near zero.
\end{lemma}

\begin{proof}
Let \(v=q(0)\) and, after shrinking, define
\begin{equation}
\begin{aligned}
a_0(t)&=\frac{v^*q(t)}{\|v\|^2}\ne0,
&z(t)&=\frac{P_v^\perp q(t)}{a_0(t)},\\
q(t)&=a_0(t)(v+z(t)).
\end{aligned}
\end{equation}
For two chart offsets \(z_1,z_2\perp v\), put \(\widehat v=v/\|v\|\),
\(\zeta=z_1/\|v\|\), and \(\xi=z_2/\|v\|\). Since \(\zeta,\xi\perp\widehat v\),
direct expansion gives
\begin{equation*}
d_{\rm ch}([\widehat v+\zeta],[\widehat v+\xi])^2
=\frac{\|\zeta-\xi\|^2+\|\zeta\wedge\xi\|^2}
{(1+\|\zeta\|^2)(1+\|\xi\|^2)}.
\end{equation*}
Consequently, on \(\|\zeta\|,\|\xi\|\le R_{\rm ch}\),
\begin{equation*}
\frac{\|\zeta-\xi\|}{1+R_{\rm ch}^2}
\le d_{\rm ch}([\widehat v+\zeta],[\widehat v+\xi])
\le\sqrt{1+R_{\rm ch}^2}\,\|\zeta-\xi\|,
\end{equation*}
where \(\|\zeta\wedge\xi\|=\|\zeta\wedge(\xi-\zeta)\|\le R_{\rm ch}\|\xi-\zeta\|\). Thus one common
chart shrink gives uniform metric constants.
The following subanalytic facts are used \cite{Tamm1981,Parusinski1994,Valette2025}: a proper
fibrewise minimum of a globally subanalytic function is globally subanalytic; a positive
one-variable subanalytic germ has a Puiseux leading term \(C\rho^c+o(\rho^c)\), \(C>0\),
\(c\in\mathbb Q\); and a nonempty definable fibre admits a definable choice arc. Thus
\eqref{eq:supp-deltachi} is either eventually zero or has such a leading power.

If the projective germ is nonconstant, write
\begin{equation}
z(t)=z(0)+z_{k_0}t^{k_0}+O(t^{k_0+1}),\qquad z_{k_0}\ne0.
\label{eq:supp-firstjet}
\end{equation}
Then
\begin{equation}
z(t)-z(-s)=z_{k_0}\big(t^{k_0}-(-1)^{k_0}s^{k_0}\big)+O((t+s)^{k_0+1}).
\label{eq:supp-parity}
\end{equation}
Equal radii show that every finite contact exponent is at least \(k_0\). If \(k_0\) is odd,
\(t^{k_0}+s^{k_0}\asymp(t+s)^{k_0}\) on every comparison cone, so every window has exponent
\(k_0\). If \(k_0\) is even, an arc of order \(>k_0\) must have \(s/t\to1\); arcs bounded away from
that ratio have order \(k_0\). For \(1<\chi_1<\chi_2\), monotonicity gives one exponent inequality.
A minimizing arc in the \(\chi_2\)-window either has order \(k_0\), available in the
\(\chi_1\)-window, or satisfies \(s/t\to1\); rebasing by \(\widehat\rho=\min(t,s)\) puts it
eventually in every fixed \(\chi_1>1\). This proves window independence. If \(z\) is constant,
every window has exact equal-radius collision.

When \(c=\infty\), definable choice yields positive comparable arcs with \([q(t)]=[q(-s)]\).
Equation \eqref{eq:supp-parity} forces \(s/t\to1\) unless the projective germ is constant, in which
case take \(s=t\). Normalize a nonzero coordinate to write
\begin{equation}
q(t)=\theta q(-s),\qquad \theta\to1.
\end{equation}
For any \(\alpha_0\) in the interior of the strict annulus, set
\begin{equation}
\beta=\alpha_0(-1)^r(t/s)^r\theta.
\end{equation}
Then \(|\beta|/|\alpha_0|\to1\), so \(\beta\) is eventually admissible, and \(\alpha_0h(t)=\beta
h(-s)\).
\end{proof}

\begin{lemma}[Full-pair radial--contact law]
Under the finite branch of Lemma~\ref{lem:supp-contact}, with a fixed strict annulus and \(p=r+c\),
the uniform full-pair exponent is exactly \(p\). Moreover the all-scale profile in the main paper
satisfies \(\Delta_{\mathcal A}(\rho)=C_{\mathcal A}\rho^p+o(\rho^p)\), \(C_{\mathcal A}>0\). The
fixed-magnitude witness in the main paper has exponent \(r+3\), rather than the strict-annulus
exponent \(r+5\).
\end{lemma}

\begin{proof}
Shrink the chart so that \(0<n_-\le\|q(t)\|\le n_+\). Choose \(0<\kappa<1/2\) satisfying
\begin{equation}
A_+n_+\kappa^r\le\tfrac12A_-n_-.
\label{eq:supp-kappa}
\end{equation}
Fix this \(\kappa\) once and assign each \((t,u)\in[-\eps,\eps]^2\) to the first applicable set
\begin{align*}
\mathsf C_0&=\{tu=0\},\\
\mathsf C_1&=\{tu\ne0,\ \min(|t|,|u|)\le\kappa M\},\\
\mathsf C_2&=\{tu>0,\ \min(|t|,|u|)>\kappa M\},\\
\mathsf C_3&=\{tu<0,\ \min(|t|,|u|)>\kappa M\},
\end{align*}
where \(M=\max(|t|,|u|)\). These sets are disjoint under the priority convention and exhaustive,
including oscillatory radius ratios. In \(\mathsf C_2\cup\mathsf C_3\), the radii differ by at most
\(\chi=\kappa^{-1}\).

For nonzero \(x,y\), completion of the square gives
\begin{equation*}
\begin{aligned}
\|\alpha x-\beta y\|^2
&=\|y\|^2\left|\beta-\alpha\frac{y^*x}{\|y\|^2}\right|^2\\
&\quad+|\alpha|^2\left(\|x\|^2-\frac{|y^*x|^2}{\|y\|^2}\right).
\end{aligned}
\end{equation*}
Hence
\begin{equation*}
\inf_{\beta\in\C}\|\alpha x-\beta y\|
=|\alpha|\,\|x\|d_{\rm ch}([x],[y]),
\qquad \beta_*=\alpha\frac{y^*x}{\|y\|^2}.
\end{equation*}
Every admissible \(\beta\) therefore has residual at least the unrestricted line distance;
substituting \(x=h(t)\), \(y=h(-s)\) also gives \eqref{eq:supp-betastar} with its parity and radius
factors.

If \(tu=0\), the nonzero mean has norm at least \(A_-n_-|t-u|^r\), which dominates a constant times
\(|t-u|^{r+c}\) locally. If \(M=\max(|t|,|u|)\) and \(\min(|t|,|u|)\le\kappa M\), the reverse
triangle inequality and \eqref{eq:supp-kappa} give
\begin{equation}
\|\alpha h(t)-\beta h(u)\|\ge\tfrac12A_-n_-M^r
\gtrsim |t-u|^{r+c}.
\end{equation}

For balanced same-side pairs, choose a real linear functional \(\ell\) with \(\ell(z_k)>0\).
Equation \eqref{eq:supp-firstjet} makes \(\ell\circ z\) strictly monotone on each half-axis and
gives, uniformly,
\begin{equation}
d_{\rm ch}([q(t)],[q(u)])\gtrsim M^{k_0-1}|t-u|.
\end{equation}
The exact line-profile identity then yields
\begin{equation}
\|\alpha h(t)-\beta h(u)\|
\gtrsim M^{r+k_0-1}|t-u|
\gtrsim|t-u|^{r+k_0}\ge|t-u|^{r+c},
\end{equation}
because \(c\ge k_0\) and \(0<|t-u|<1\).

For balanced opposite-side pairs, write \(u=-s\), \(t,s>0\). Balance places \((t,s)\) in one fixed
contact window. The prepared positive function gives
\begin{equation}
d_{\rm ch}([q(t)],[q(-s)])\gtrsim(t+s)^c.
\end{equation}
Profiling over all complex gains can only reduce the distance from an admissible competitor; hence,
by the line-profile identity,
\begin{equation}
\begin{aligned}
\|\alpha h(t)-\beta h(-s)\|
&\ge |\alpha|\,\|h(t)\|d_{\rm ch}([q(t)],[q(-s)])\\
&\gtrsim(t+s)^{r+c}.
\end{aligned}
\end{equation}
The minimum of the finitely many constants proves the full-pair lower bound.

All estimates use one common \(0<\eps\le\eps_0\) for which the chart, monotonicity, preparation,
and \(2\eps<1\) bounds hold. If \(C_j>0\) is the constant on \(\mathsf C_j\), then
\(C_\eps=\min_jC_j>0\) applies to every pair: for \(0\le d<1\), \(d^r\ge d^{r+c}\) and
\(d^{r+k_0}\ge d^{r+c}\) because \(c\ge k_0\). A sequence crossing cases cannot destroy uniformity.

For sharpness, the unrestricted least-squares gain equals
\begin{equation}
\beta_*(t,s;\alpha)=\alpha(-1)^r(t/s)^r
\frac{q(-s)^*q(t)}{\|q(-s)\|^2}.
\label{eq:supp-betastar}
\end{equation}
A contact-minimizing arc of order \(c>k_0\) has \(s/t\to1\); if \(c=k_0\), equal radii provide an
arc of the same order. In either case one may take a sharp arc with \(|\beta_*|/|\alpha|\to1\). An
interior true gain therefore makes \(\beta_*\) admissible, and the separation is \(O((t+s)^{r+c})\).

The compact minimum \(\Delta_{\mathcal A}\) is positive at every sufficiently small scale, or
finite-branch injectivity would fail. It is a one-variable globally subanalytic function, so
\(\Delta_{\mathcal A}=C_{\mathcal A}\rho^{p_\Delta}+o(\rho^{p_\Delta})\). The lower bound gives
\(p_\Delta\le p\). Rebase the sharp arc by \(\rho=\min(t,s)\); its ratio tends to one, so both
radii lie in \([\rho,2\rho]\), giving \(p_\Delta\ge p\). Thus \(p_\Delta=p\).

Finally take \(h(t)=t^r(1,t^2+t^3,t^5)^{\mathsf T}\). In the strict annulus, put \(s=t\upsilon(t)\)
and solve
\begin{equation}
\upsilon^2-t\upsilon^3=1+t,\qquad
\upsilon=1+t+t^2+2t^3+O(t^4).
\end{equation}
The second coordinate is then matched by a gain-magnitude adjustment and the remaining projective
residual is \(t^5+s^5=\Theta(t^5)\), giving \(r+5\). Under fixed magnitude, an \(o(t^3)\) opposite
residual would force \(s/t=1+o(t^3)\) and unit phase \(1+o(t^3)\), but the second coordinate would
remain \(2t^3+o(t^3)\). Equal radii attain order \(r+3\) with \(\beta=(-1)^r\alpha\).
\end{proof}

\section{Local Signed Minimax Conversion}

\begin{lemma}[Gaussian modulus-to-risk conversion]
For the fixed-neighborhood whitened experiment of the main paper, finite \(p=r+c\) implies local
signed minimax mean-square risk \(\Theta(R^{-1/p})\) for every sufficiently large integer \(R\).
Infinite contact prevents consistency.
\end{lemma}

\begin{proof}
For \(Y\sim\CN(\mu,\sigma^2I_m)\), direct integration gives
\begin{equation}
\operatorname{KL}(P_\mu\|P_\nu)=\frac{\|\mu-\nu\|^2}{\sigma^2}.
\end{equation}
Choose \(M_0,\rho_0>0\) with \(\Delta_{\mathcal A}(\rho)\le M_0\rho^p\) for \(0<\rho\le\rho_0\).
Fix \(0<\eta<1\), take \(a>0\) so \(M_0a^p/(\sigma\sqrt2)\le\eta\), and for every large \(R\) put
\(\rho_R=aR^{-1/(2p)}\). An attained profile minimizer has signed separation
\(d_R\in[2\rho_R,4\rho_R]\) and \(R\)-sample divergence at most \(M_0^2a^{2p}/\sigma^2\). The
chosen inequality makes the total KL at most \(2\eta^2\), so Pinsker gives total variation at most
\(\eta\). If \(A\) is the event on which an estimator is closer to the first parameter point, its
squared error is at least \(d_R^2/4\) on \(A^c\) under the first law and on \(A\) under the second.
Therefore
\begin{equation*}
\max_{i=0,1}E_i(\widehat t-t_i)^2
\ge\frac{d_R^2}{8}\{P_0(A^c)+P_1(A)\}
\ge\frac{d_R^2}{8}(1-\operatorname{TV}),
\end{equation*}
which supplies the displayed constant below.
\begin{equation}
\mathcal R_R(\eps)\ge\frac{d_R^2}{8}(1-\eta)
\ge\frac{1-\eta}{2}a^2R^{-1/p}.
\end{equation}
The all-scale profile makes this an all-large-\(R\) statement, not a subsequence.

For the upper bound, the sufficient sample mean is
\begin{equation}
\bar Y=\alpha h(t)+\bar W,\qquad
\bar W\sim\CN(0,\sigma^2I_m/R).
\end{equation}
Sufficiency follows from
\begin{equation*}
\sum_{j=1}^R\|Y_j-\mu\|^2
=\sum_{j=1}^R\|Y_j-\bar Y\|^2+R\|\bar Y-\mu\|^2,
\end{equation*}
whose first term is independent of \(\mu\). Let \(\Theta=I_\eps\times\mathcal A\), a compact metric
space, and \(f(y;u,\beta)=\|y-\beta h(u)\|\). Its continuous value \(v(y)=\min_\Theta f(y;\cdot)\)
gives a Borel argmin graph with nonempty compact sections; the stated selection theorem yields a
Borel selector without randomization.
Let \((\widehat t,\widehat\alpha)\) denote that selector. Its optimality and the triangle
inequality give
\begin{equation}
\|\widehat\alpha h(\widehat t)-\alpha h(t)\|\le2\|\bar W\|.
\end{equation}
The full-pair modulus then implies
\begin{equation}
|\widehat t-t|^2\le(2/C_\eps)^{2/p}\|\bar W\|^{2/p}.
\end{equation}
For \(Z\sim\CN(0,I_m)\), \(\|Z\|^2\sim\operatorname{Gamma}(m,1)\), whence
\begin{equation}
\mathbb E\|Z\|^{2/p}=\frac{\Gamma(m+1/p)}{\Gamma(m)}.
\end{equation}
Taking expectations proves
\begin{equation}
\mathcal R_R(\eps)\le
(2/C_\eps)^{2/p}\sigma^{2/p}
\frac{\Gamma(m+1/p)}{\Gamma(m)}R^{-1/p}.
\end{equation}
If \(c=\infty\), Lemma~\ref{lem:supp-contact} supplies in every fixed interval a pair
\(t_\eps,-s_\eps\) separated by \(d_\eps>0\) with the same law \(P\). With midpoint \(\bar
t_\eps=(t_\eps-s_\eps)/2\), every estimator satisfies
\begin{equation*}
\frac12E_P\!\left[(\widehat t-t_\eps)^2+(\widehat t+s_\eps)^2\right]
=E_P(\widehat t-\bar t_\eps)^2+\frac{d_\eps^2}{4}
\ge\frac{d_\eps^2}{4}.
\end{equation*}
The constant midpoint estimator attains equality for the two-point subproblem.
\end{proof}

\section{Spectral Extremality and Equality Cases}

\begin{lemma}[All-even extremum, rigidity, and losslessness]
The value, unique projector, rigidity inequality, regular-GI attainment, and lossless-rank
classification stated in Theorem~4 and Corollary~6 of the main paper hold.
\end{lemma}

\begin{proof}
Taking \(\lambda=0\), then averaging and using Fourier orthogonality, gives
\begin{equation}
F(P)\le\frac1{2\pi S_L}\int_0^{2\pi}\|Pg(\omega)\|^2d\omega
=\frac{\tr(PD^2)}{S_L}.
\end{equation}
For \(m=2k\), Ky Fan's theorem \cite{KyFan1949} selects the \(k\) complete symmetric pairs indexed
by \(E_{L,m}\). At \(Q=P_E\),
\begin{equation}
(Qa)^*Qg=i\sum_{n\in E}(n-\mu)=0,
\end{equation}
so profiling removes no energy and the ceiling is attained pointwise. Exact summation gives
\begin{equation}
2\sum_{j=0}^{k-1}(\mu-j)^2
=\frac{k(3L^2-6kL+4k^2-1)}6,
\end{equation}
which, after division by \(S_L\) and \(m=2k\), is the main formula.

If \(m<L\), the smallest selected minus largest unselected eigenvalue is
\begin{equation}
\left(\frac{L-m+1}{2}\right)^2-
\left(\frac{L-m-1}{2}\right)^2=L-m.
\end{equation}
For equal-rank projectors, the spectral budget therefore yields
\begin{equation}
\tr(QD^2)-\tr(PD^2)
\ge(L-m)\tr(P(I-Q)),
\end{equation}
proving rigidity. Equality forces all projector mass into the strictly separated selected
eigenspace, hence \(P=Q\). For \(m\ge4\), \(Q\) retains an adjacent pair; their ratio is
\(z=e^{i\omega}\), so \(Q\) is non-dark, injective, and immersive.

For losslessness, \(\rho_P(\omega)\le\|Pg(\omega)\|^2/S_L\le1\). If \(F(P)=1\), both inequalities
are equalities for every \(\omega\), so \(Pg(\omega)=g(\omega)\). Fourier uniqueness gives
\begin{equation}
\operatorname{span}_{\omega}\{g(\omega)\}=
\begin{cases}
\C^L,&L\ \text{even},\\
\operatorname{span}\{e_n:n\ne(L-1)/2\},&L\ \text{odd}.
\end{cases}
\end{equation}
Compactness and continuity ensure that \(\Gamma=1\) is attained, so a rank-\(m\) row space must
contain this span. This proves the iff classification and uniqueness. At \(L=3,m=2\), central
deletion leaves \([1:z^2]\), which is non-dark and immersive but aliased.
\end{proof}

\begin{proposition}[Even-aperture co-rank-one equality]
For even \(L\ge4\), the co-rank-one value and its two central-deletion regular-GI equality cases
are exactly those in Proposition~5 of the main paper.
\end{proposition}

\begin{proof}
Write \(P=I-vv^*\), \(\|v\|=1\), and \(A=v^*a\), \(G=v^*g\). A dark design has \(F=0\) and cannot
maximize. Direct scalar profiling at a non-dark design gives
\begin{equation}
F(I-vv^*)=1-\frac{L}{S_L}
\sup_\omega\frac{|G(\omega)|^2}{L-|A(\omega)|^2}.
\label{eq:supp-corankprofile}
\end{equation}
Since \(L-|A|^2>0\), the supremum is at least the ratio of Haar integrals:
\begin{equation}
\sup_\omega\frac{|G|^2}{L-|A|^2}
\ge\frac{\int|G|^2}{\int(L-|A|^2)}
=\frac{v^*D^2v}{L-1}\ge\frac1{4(L-1)}.
\end{equation}
The last equality uses Fourier orthogonality; the last inequality uses the two central eigenvalues
\(1/4\). Equality forces \(v=c_0e_{L/2-1}+c_1e_{L/2}\). Put \(a_0=2|c_0c_1|\). As \(\omega\)
varies, a real cross term \(u\) spans \([-a_0,a_0]\) and
\begin{equation}
|A|^2=1+u,\qquad |G|^2=\tfrac14(1-u).
\end{equation}
The ratio is decreasing in \(u\), hence
\begin{equation}
\begin{aligned}
\sup_\omega\frac{|G|^2}{L-|A|^2}
&=\frac{1+a_0}{4(L-1+a_0)}\\
&=\frac1{4(L-1)}+
\frac{a_0(L-2)}{4(L-1)(L-1+a_0)}.
\end{aligned}
\end{equation}
Equality holds iff \(a_0=0\), i.e., precisely one central coefficient is nonzero. Substitution into
\eqref{eq:supp-corankprofile} yields
\begin{equation}
\Gamma_{L,L-1}=1-\frac{3}{(L-1)^2(L+1)}.
\end{equation}
Each deletion retains an adjacent coordinate pair (for \(L=4\), deleting coordinate \(1\) leaves
\(2,3\)), proving regular GI.
\end{proof}

\section{Winding, Incidence, and Continuity}

\begin{lemma}[Rank-two winding chamber]
Every rank-two GI projector satisfies
\begin{equation}
\sin^2\theta_{\max}\ge\tau_L:=\frac{4L-1}{4L^2},
\end{equation}
where \(\theta_{\max}\) is its largest principal angle from the endpoint plane. Consequently the
price floor in Theorem~7 holds.
\end{lemma}

\begin{proof}
Align endpoint and candidate frames \(U=[e_0,e_{L-1}]\) and \(V=[v_0,v_1]\) by the polar factor so
that \(U^*V\succeq0\). Principal-angle calculus gives
\begin{equation}
\|V-U\|_{\rm op}=2\sin(\theta_{\max}/2).
\end{equation}
Write
\begin{equation}
\begin{aligned}
f_0(z)&=v_0^*a(z)=1+r_0(z),\\
f_1(z)&=v_1^*a(z)=z^{L-1}+r_1(z).
\end{aligned}
\end{equation}
Cauchy--Schwarz on the coefficient vectors gives
\begin{equation}
\|r_j\|_{L^\infty(S^1)}
\le2\sqrt L\sin(\theta_{\max}/2).
\end{equation}
For principal angles in \([0,\pi/2]\), the right side is \(<1\) exactly under the strict condition
\(\sin^2\theta_{\max}<\tau_L\). Rouch\'e's theorem \cite{Conway1978} then gives no disk zeros for
\(f_0\), \(L-1\) disk zeros for \(f_1\), and no zeros on \(S^1\). Thus \(f_1/f_0:S^1\to\C^\times\)
has winding \(L-1\). If the projective response were injective, this ratio would be a Jordan loop.
The bounded component of a Jordan loop contains either the origin, giving index \(\pm1\), or not,
giving index \(0\); hence an injective loop in \(\C^\times\) cannot have index \(L-1\ge2\). The
chamber is non-GI.

If \(\theta_1,\theta_2\) are the two principal angles, then
\begin{equation}
\tr(P(I-Q))=\sin^2\theta_1+\sin^2\theta_2
\ge\sin^2\theta_{\max}.
\end{equation}
The rank-two specialization of rigidity is
\begin{equation}
F(P)\le\Gamma_{L,2}-\frac{L-2}{S_L}\tr(P(I-Q)).
\end{equation}
Every GI plane therefore loses at least \((L-2)\tau_L/S_L\), which simplifies to
\begin{equation}
\Pi_{L,2}\ge\frac{3(L-2)(4L-1)}{L^3(L^2-1)}.
\end{equation}
The chamber inequality is strict; neither its boundary nor the exact constrained optimum is
identified. An adjacent-coordinate plane shows GI feasibility and has \(F=1/(2S_L)\).
\end{proof}

\begin{lemma}[Compactified physical incidence]\label{lem:supp-incidence}
For \(3\le m<L\), regular-GI projectors form an open dense subset of \(\Gr(m,L)\). The full-rank
case is trivially regular-GI.
\end{lemma}

\begin{proof}
For \(z,w\in S^1\), define \(b_0=0\) and
\begin{equation}
\begin{aligned}
b_n(z,w)&=\frac{w^n-z^n}{w-z}
=\sum_{j=0}^{n-1}w^{n-1-j}z^j,\\
b(z,z)&=\partial_za(z).
\end{aligned}
\label{eq:supp-divided}
\end{equation}
Because \(a\) has first coordinate \(1\) and \(b\) has first two coordinates \((0,1)\),
\(
\widetilde S(z,w)=\operatorname{span}_{\C}\{a(z),b(z,w)\}
\)
is always a two-plane. Moreover
\begin{equation}
Pa(z)\wedge Pb(z,w)=0
\quad\Longleftrightarrow\quad
\ker P\cap\widetilde S(z,w)\ne\{0\}.
\label{eq:supp-wedge}
\end{equation}
For \(z\ne w\), this is an off-diagonal projective alias or darkness; for \(z=w\), it is darkness
or projective criticality. Formula \eqref{eq:supp-divided} is the polynomial representative of the
divided wedge on the real-oriented blow-up of the diagonal, so the augmented zero set is closed on
a compact design--pair product. This is a joint design--pair closure, not a fixed-design converse.

Put \(k_{\rm ker}=L-m\) and fix a complex two-plane \(S\). Consider
\begin{equation*}
\mathcal I_S=\{(\ell,K):\ell\in\mathbb P(S),\ K\in\Gr(k_{\rm ker},L),\ \ell\subset K\}.
\end{equation*}
For fixed \(\ell\), the fibre is \(\Gr(k_{\rm ker}-1,L-1)\), of complex dimension \((k_{\rm
ker}-1)m\). Hence
\begin{equation}
\dim_\C\mathcal I_S=1+(k_{\rm ker}-1)m=k_{\rm ker}m-(m-1).
\end{equation}
The map \(\mathcal I_S\to\Sigma_S\) is one-to-one over \(\dim(K\cap S)=1\), giving the top-stratum
dimension. If \(k_{\rm ker}\ge2\), the only exceptional stratum is \(\{K:S\subset
K\}\simeq\Gr(k_{\rm ker}-2,L-2)\), of dimension \((k_{\rm ker}-2)m\), smaller by \(m+1\); it is
empty for \(k_{\rm ker}=1\). No higher intersection is possible. Thus the fixed-pair condition has
real codimension \(2m-2\).

The ordered physical off-diagonal pair contributes two real parameters, so the incidence has real
dimension at most
\begin{equation*}
2+2\{k_{\rm ker}m-(m-1)\}=2k_{\rm ker}m-2m+4,
\end{equation*}
and its design projection has codimension at least \(2m-4\). The one-dimensional boundary gives at
least \(2m-3\). For darkness, a fixed frequency prescribes the line \(\C a(z)\subset K\); the
resulting \(\Gr(k_{\rm ker}-1,L-1)\) has complex codimension \(m\). Allowing frequency leaves
codimension at least \(2m-1\). The case \(k_{\rm ker}=0\) is full rank and was separated in the
lemma statement.

These sets are semialgebraic in the real algebraic model \(x^2+y^2=1\). Tarski--Seidenberg
projection preserves semialgebraicity and cannot increase dimension \cite{Bochnak1998}. Compactness
makes each projected augmented bad set closed. For \(m\ge3\), every bad component has positive
codimension, so the complement is dense. By \eqref{eq:supp-wedge}, the complement is exactly
non-dark, projectively immersive, and globally injective; compact-source embedding stability makes
it open. This is regular GI.
\end{proof}

\begin{lemma}[Zero set and global continuity of \(F\)]
The zero-set characterization in Proposition~3 holds, and \(F\) is continuous on the entire
Grassmannian. Hence regular-GI density implies equality of constrained and unconstrained suprema
for \(m\ge3\).
\end{lemma}

\begin{proof}
At a non-dark frequency, finite least squares gives
\begin{equation}
\rho_P(\omega)=\frac1{S_L}\left(\|Pg\|^2-
\frac{|(Pa)^*Pg|^2}{\|Pa\|^2}\right).
\label{eq:supp-pointwise}
\end{equation}
It is continuous there and vanishes exactly when \(Pg\) is collinear with \(Pa\), i.e., at a
projective critical point. If \(Pa(\omega)\ne0\) everywhere, compactness gives
\(\inf_\omega\|Pa(\omega)\|>0\), so \eqref{eq:supp-pointwise} is continuous on the entire circle.
It is pointwise positive exactly for a projective immersion, and compactness then makes its infimum
positive. This argument is not used at darkness.

At a dark point, write \(h(t)=Pa(\omega_0+t)=t^rq(t)\). Since \(Pg=h'-i\mu h\), choosing
\(\lambda_t=r/t\) gives
\begin{equation}
P(g-\lambda_ta)=t^r(q'-i\mu q)\longrightarrow0,
\label{eq:supp-darklimit}
\end{equation}
so \(F(P)=0\). This proves the iff zero set. It also shows why the pointwise algebraic profile may
jump and why the frequency infimum need not be attained.

For continuity, define
\begin{equation}
J(P,\omega,\lambda)=S_L^{-1}
\|P(g(\omega)-\lambda a(\omega))\|^2.
\end{equation}
Let \(P_j\to P\) be arbitrary. For every \(\varepsilon>0\), the infimum supplies finite
\((\omega_\varepsilon,\lambda_\varepsilon)\) with \(J(P,\omega_\varepsilon,\lambda_\varepsilon)\le
F(P)+\varepsilon\). Hence
\begin{equation*}
\limsup_jF(P_j)\le\lim_jJ(P_j,\omega_\varepsilon,\lambda_\varepsilon)
\le F(P)+\varepsilon.
\end{equation*}
Letting \(\varepsilon\downarrow0\) proves upper semicontinuity without assuming attainment. For
lower semicontinuity, pass to a subsequence for which \(F(P_j)\to\ell=\liminf_iF(P_i)\), and choose
\begin{equation}
J(P_j,\omega_j,\lambda_j)\le F(P_j)+j^{-1}.
\end{equation}
Since \(0\le F(P_j)\le1\), the residuals are bounded. Compactness gives \(\omega_j\to\omega\) along
a subsequence. If \((\lambda_j)\) is bounded, pass to \(\lambda_j\to\lambda\); joint continuity
gives \(F(P)\le J(P,\omega,\lambda)=\ell\). Otherwise pass to a subsequence with
\(|\lambda_j|\to\infty\). Bounded residuals and \(\|P_jg(\omega_j)\|\le\sqrt{S_L}\) imply
\begin{equation}
P_ja(\omega_j)=
\frac{P_jg(\omega_j)-P_j(g(\omega_j)-\lambda_ja(\omega_j))}{\lambda_j}
\longrightarrow0.
\end{equation}
Hence \(Pa(\omega)=0\), and \eqref{eq:supp-darklimit} gives \(F(P)=0\le\ell\). Thus
\(F(P)\le\liminf_jF(P_j)\) for every sequence, proving continuity on the metric Grassmannian.

The compact Grassmannian has a maximizer \(P_\star\). Lemma~\ref{lem:supp-incidence} supplies
regular-GI \(P_j\to P_\star\) for \(m\ge3\); continuity gives \(F(P_j)\to\Gamma_{L,m}\). Since
every regular-GI projector is GI,
\begin{equation}
\Gamma^{\GI}_{L,m}=\Gamma_{L,m},\qquad 3\le m\le L,
\end{equation}
as an equality of suprema. No general odd-rank attaining projector is inferred.
\end{proof}

\bibliographystyle{IEEEtran}
\bibliography{references}

\end{document}